\documentclass[11pt, fleqn]{article}
\usepackage[english]{babel}
\usepackage[utf8]{inputenc}
\usepackage[T1]{fontenc}
\usepackage{amssymb}
\usepackage{color}
\usepackage[lmargin=1.1in,rmargin=1.1in,bottom=1.3in,top=1.3in,twoside=False]{geometry}

\usepackage{amssymb}
\usepackage{microtype}
\usepackage{amsmath}
\usepackage{amssymb}
\usepackage{amsfonts}
\usepackage{mathtools}

\usepackage{mysavetreesbib}

\usepackage{latexsym}

\usepackage[urlcolor=blue, linkcolor=blue, citecolor=blue, colorlinks=true]{hyperref}
\usepackage[amsmath,thmmarks, hyperref]{ntheorem}

\theoremnumbering{arabic}
\theoremstyle{plain}
\theoremsymbol{}
\theorembodyfont{\itshape}
\theoremheaderfont{\normalfont\bfseries}
\theoremseparator{.}

\newtheorem{theorem}{Theorem}

\newtheorem{lemma}[theorem]{Lemma}

\newtheorem{corollary}[theorem]{Corollary}
\theorembodyfont{\upshape}

\theoremstyle{nonumberplain}
\newtheorem{claim}{Claim}

\theoremsymbol{\ensuremath{□}}

\theoremstyle{nonumberplain}
\theoremheaderfont{\scshape}
\theorembodyfont{\normalfont}
\theoremsymbol{\ensuremath{\square}}
\theoremseparator{}
\newtheorem{proof}{Proof.}

\def\cqedsymbol{\ifmmode$\lrcorner$\else{\unskip\nobreak\hfil
\penalty50\hskip1em\null\nobreak\hfil$\lrcorner$
\parfillskip=0pt\finalhyphendemerits=0\endgraf}\fi}

\theoremsymbol{\ensuremath{\dashv}}
\newtheorem{claimproof}{Proof.}

\newcommand{\wcol}{\mathrm{wcol}}
\newcommand{\col}{\mathrm{col}}
\newcommand{\adm}{\mathrm{adm}}
\newcommand{\Wreach}{\mathrm{WReach}}
\newcommand{\Sreach}{\mathrm{SReach}}

\newcommand{\Oof}{\mathcal{O}}
\newcommand{\CCC}{\mathcal{C}}

\newcommand{\N}{\mathbb{N}}
\renewcommand{\phi}{\varphi}
\newcommand{\strA}{\mathfrak{A}}
\newcommand{\FO}{\mathrm{FO}}
\newcommand{\minor}{\preccurlyeq}
\newcommand{\dist}{\mathrm{dist}}

\title{The Generalised Colouring Numbers on \\Classes of Bounded Expansion\footnote{This work was initiated during Sebastian Siebertz's visit 
at the Institute of Informatics of the University of Warsaw,
which was supported by the Warsaw Centre of Mathematics and Computer Science. Micha\l{} Pilipczuk is supported by the Foundation for Polish Science (FNP) via the START stipend programme. Stephan Kreutzer, Roman Rabinovich and Sebastian Siebertz's research has been supported by the
European Research Council (ERC) under the European Union’s Horizon
2020 research and innovation programme (ERC Consolidator Grant
DISTRUCT, grant agreement No 648527).}}
\author{Stephan Kreutzer, Michał Pilipczuk, Roman Rabinovich, and Sebastian Siebertz}

\begin{document}

\maketitle

\begin{abstract}
  \noindent The generalised colouring numbers $\adm_r(G)$,
  $\col_r(G)$, and $\wcol_r(G)$ were introduced by Kierstead and Yang
  as generalisations of the usual colouring number, also known as the
  degeneracy of a graph, and have since then found important
  applications in the theory of bounded expansion and nowhere dense
  classes of graphs, introduced by Ne\v{s}et\v{r}il and Ossona de
  Mendez. In this paper, we study the relation of the colouring
  numbers with two other measures that characterise nowhere dense
  classes of graphs, namely with uniform quasi-wideness, studied first
  by Dawar et al.\ in the context of preservation theorems for
  first-order logic, and with the splitter game, introduced by Grohe
  et al. We show that every graph excluding a fixed topological minor
  admits a universal order, that is, one order witnessing that the
  colouring numbers are small for every value of $r$. Finally, we use
  our construction of such orders to give a new proof of a result of
  Eickmeyer and Kawarabayashi, showing that the model-checking problem
  for successor-invariant first-order formulas is fixed-parameter
  tractable on classes of graphs with excluded topological minors.
\end{abstract}

\section{Introduction}

The \emph{colouring number $\col(G)$} of a graph $G$ is the minimum
$k$ for which there is a linear order~$<_L$ on the vertices of $G$
such that each vertex $v$ has \emph{back-degree} at most $k-1$, that
is, $v$ has at most $k-1$ neighbours $u$ with $u<_Lv$. The colouring
number is a measure for uniform sparseness in graphs: we have
$\col(G)=k$ if and only if every subgraph $H$ of $G$ has a vertex of
degree at most $k-1$. Hence, provided $\col(G)=k$, not only $G$ is
sparse, but also every subgraph of $G$ is sparse. The colouring number
minus one is also known as the \emph{degeneracy}.

Recently, Ne\v{s}et\v{r}il and Ossona de Mendez introduced the notions
of \emph{bounded expansion}~\cite{nevsetvril2008grad} and
\emph{nowhere density} \cite{nevsetvril2011nowhere} as very general
formalisations of uniform sparseness in graphs. Since then, several
independent and seemingly unrelated characterisations of these notions
have been found, showing that these concepts behave robustly.  For
example, nowhere dense classes of graphs can be defined in terms of
excluded shallow minors~\cite{nevsetvril2011nowhere}, in terms of
uniform quasi-wideness~\cite{dawar2010homomorphism}, a notion studied
in model theory, or in terms of a game~\cite{grohe2014deciding} with
direct algorithmic applications. The \emph{generalised colouring
  numbers} $\adm_r$, $\col_r$, and $\wcol_r$ were introduced by
Kierstead and Yang~\cite{kierstead2003orders} in the context of
colouring and marking games on graphs. As proved by Zhu
\cite{zhu2009colouring}, they can be used to characterise both bounded
expansion and nowhere dense classes of graphs.

The invariants $\adm_r$, $\col_r$, and $\wcol_r$ are defined similarly
to the classic colouring number: for example, the \emph{weak
  $r$-colouring} number $\wcol_r(G)$ of a graph $G$ is the minimum
integer~$k$ for which there is a linear order of the vertices such
that each vertex $v$ can reach at most $k-1$ vertices $w$ by a path of
length at most~$r$ in which~$w$ is the smallest vertex on the
path. 

The generalised colouring numbers found important applications in the
context of algorithmic theory of sparse graphs. For example, they play
a key role in Dvo\v{r}\'ak's approximation algorithm for minimum
dominating sets \cite{dvovrak13}, or in the construction of sparse
neighbourhood covers on nowhere dense classes, a fundamental step in
the almost linear time model-checking algorithm for first-order
formulas of Grohe et al.~\cite{grohe2014deciding}.
  
In this paper we study the relation between the colouring numbers and
the above mentioned characterisations of nowhere dense classes of
graphs, namely with uniform quasi-wideness and the splitter game. We
use the generalised colouring numbers to give a new proof that every
bounded expansion class is uniformly quasi-wide.  This was first
proved by Ne\v{s}et\v{r}il and Ossona de Mendez in
\cite{nevsetvril2010first}; however, the constants appearing in the
proof of~\cite{nevsetvril2010first} are huge. We present a very simple
proof which also improves the appearing constants. Furthermore, for
the splitter game introduced in~\cite{grohe2014deciding}, we show that
splitter has a very simple strategy to win on any class of bounded
expansion, which leads to victory much faster than in general nowhere
dense classes of graphs.

Every graph $G$ from a fixed class $\CCC$ of bounded expansion
satisfies $\wcol_r(G)\leq f(r)$ for some function $f$ and all positive
integers~$r$. However, the order that witnesses this inequality for
$G$ may depend on the value $r$. We say that a class $\CCC$ admits
\emph{uniform orders} if there is a function $f:\N\rightarrow\N$ such
that for each $G\in\CCC$ there is one linear order that witnesses
$\wcol_r(G)\leq f(r)$ for every value of $r$.  We show that every
class that excludes a fixed topological minor admits uniform orders
that can be computed efficiently.

Finally, based on our construction of uniform orders for graphs that
exclude a fixed topological minor, we provide an alternative proof of
a very recent result of Eickmeyer and
Kawarabayashi~\cite{eickmeyer2016model}, that the model-checking
problem for successor-invariant first-order ($\FO$) formulas is
fixed-para\-meter tractable on such classes (we obtained this result
independently of but later than~\cite{eickmeyer2016model}).
Successor-invariant logics have been studied in database theory and
finite model theory, and successor-invariant $\FO$ is known to be more
expressive than plain $\FO$~\cite{rossman2007successor}.  The
model-checking problem for successor-invariant $\FO$ is known to be
fixed-parameter tractable parameterized by the size of the formula on
any graph class that excludes a fixed minor~\cite{eickmeyer2013model}.
Very recently, this result was lifted to classes that exclude a fixed
topological minor by Eickmeyer and
Kawarabayashi~\cite{eickmeyer2016model}. 
The key point of their proof
is to use the decomposition theorem for graphs excluding a fixed
topological minor, due to Grohe and Marx~\cite{grohe2015structure}.
Our approach is similar to that of~\cite{eickmeyer2016model}.
However, we employ new constructions based on the generalised
colouring numbers and use the decomposition theorem
of~\cite{grohe2015structure} only implicitly. In particular, we do not
construct a graph decomposition in order to solve the model-checking
problem. Therefore, we believe that our approach may be easier to
extend further to classes of bounded expansion, or even to nowhere
dense classes of graphs.

\section{Preliminaries}

\subparagraph*{Notation.} We use standard graph-theoretical notation;
see e.g.~\cite{diestel2012graph} for reference.  All graphs considered
in this paper are finite, simple, and undirected.  For a graph $G$, by
$V(G)$ and $E(G)$ we denote the vertex and edge sets of $G$,
respectively.  A graph~$H$ is a \emph{subgraph} of~$G$, denoted
$H\subseteq G$, if $V(H)\subseteq V(G)$ and $E(H)\subseteq E(G)$.  For
any $M\subseteq V(G)$, by $G[M]$ we denote the subgraph induced by
$M$.  We write $G-M$ for the graph $G[V(G)\setminus M]$ and if
$M=\{v\}$, we write $G-v$ for $G-M$. For a non-negative integer
$\ell$, a \emph{path of length $\ell$} in~$G$ is a sequence
$P=(v_1,\ldots, v_{\ell+1})$ of pairwise different vertices such that
$v_iv_{i+1}\in E(G)$ for all $1\leq i\leq \ell$. We write $V(P)$ for
the vertex set $\{v_1,\ldots, v_{\ell+1}\}$ of~$P$ and $E(P)$ for the
edge set $\{v_iv_{i+1} : 1\leq i\leq \ell\}$ of~$P$ and identify~$P$
with the subgraph of $G$ with vertex set $V(P)$ and edge set
$E(P)$. We say that the path $P$ \emph{connects} its \emph{endpoints}
$v_1,v_{\ell+1}$, whereas $v_2,\ldots, v_\ell$ are the \emph{internal
  vertices} of $P$. The {\em{length}} of a path is the number of its
edges.  Two vertices $u,v\in V(G)$ are \emph{connected} if there is a
path in $G$ with endpoints $u,v$.  The {\em{distance}} $\dist(u,v)$
between two connected vertices $u,v$ is the minimum length of a path
connecting $u$ and $v$; if $u,v$ are not connected, we put
$\dist(u,v)=\infty$. The \emph{radius} of~$G$ is
$\min_{u\in V(G)}\max_{v\in V(G)}\dist(u,v)$. The set of all
neighbours of a vertex $v$ in $G$ is denoted by $N^G(v)$, and the set
of all vertices at distance at most $r$ from $v$ is denoted by
$N^G_r(v)$.  A graph $G$ is $c$-\emph{degenerate} if every subgraph
$H\subseteq G$ has a vertex of degree at most $c$. A $c$-degenerate
graph of order $n$ contains an independent set of order at least
$n/(c+1)$.

A graph $H$ with $V(H)=\{v_1,\ldots, v_n\}$ is a \emph{minor} of $G$,
written $H\minor G$, if there are pairwise disjoint connected
subgraphs $H_1,\ldots, H_n$ of $G$, called {\em{branch sets}}, such
that whenever $v_iv_j\in E(H)$, then there are $u_i\in H_i$ and
$u_j\in H_j$ with $u_iu_j\in E(G)$. We call $(H_1,\ldots, H_n)$ a
{\em{minor model}} of $H$ in~$G$. The graph $H$ is a \emph{topological
  minor} of $G$, written $H\minor^t G$, if there are pairwise
different vertices $u_1, \ldots, u_n\in V(G)$ and a family of paths
$\{P_{ij}\ \colon\ v_iv_j\in E(H)\}$, such that each $P_{ij}$ connects
$u_i$ and $u_j$, and paths $P_{ij}$ are pairwise internally
vertex-disjoint.

\vspace{-0.5cm}
\subparagraph*{Generalised colouring numbers.} Let us fix a graph
$G$. By $\Pi(G)$ we denote the set of all linear orders of $V(G)$. For
$L\in\Pi(G)$, we write $u<_L v$ if $u$ is smaller than $v$ in $L$, and
$u\le_L v$ if $u<_L v$ or $u=v$.  Let $u,v\in V(G)$. For a
non-negative integer $r$, we say that $u$ is \emph{weakly
  $r$-reachable} from~$v$ with respect to~$L$, if there is a path $P$
of length $\ell$, $0\le\ell\le r$, connecting $u$ and~$v$ such that
$u$ is minimum among the vertices of $P$ (with respect to $L$). By
$\Wreach_r[G,L,v]$ we denote the set of vertices that are weakly
$r$-reachable from~$v$ w.r.t.\ $L$.

Vertex $u$ is \emph{strongly $r$-reachable} from $v$ with respect
to~$L$, if there is a path $P$ of length~$\ell$, $0\le\ell\le r$,
connecting $u$ and $v$ such that $u\le_Lv$ and such that all internal
vertices $w$ of~$P$ satisfy $v<_Lw$. Let $\Sreach_r[G,L,v]$ be the set
of vertices that are strongly $r$-reachable from~$v$ w.r.t.\ $L$. Note
that we have $v\in \Sreach_r[G,L,v]\subseteq \Wreach_r[G,L,v]$.

For a non-negative integer $r$, we define the \emph{weak $r$-colouring
  number $\wcol_r(G)$} of $G$ and the \emph{$r$-colouring number
  $\col_r(G)$} of $G$ respectively as follows:
\begin{eqnarray*}
\wcol_r(G)& := & \min_{L\in\Pi(G)}\:\max_{v\in V(G)}\:
\bigl|\Wreach_r[G,L,v]\bigr|,\\
\col_r(G) & := & \min_{L\in\Pi(G)}\:\max_{v\in V(G)}\:
\bigl|\Sreach_r[G,L,v]\bigr|.
\end{eqnarray*}

For a non-negative integer $r$, the \emph{$r$-admissibility}
$\adm_r[G,L, v]$ of $v$ w.r.t.\ $L$ is the maximum size~$k$ of a
family $\{P_1,\ldots,P_k\}$ of paths of length at most $r$ that start
in~$v$, end at a vertex $w$ with $w\leq_Lv$, and satisfy
$V(P_i)\cap V(P_j)=\{v\}$ for all $1\leq i< j\leq k$. As for $r>0$ we
can always let the paths end in the first vertex smaller than $v$, we
can assume that the internal vertices of the paths are larger
than~$v$.  Note that $\adm_r[G,L,v]$ is an integer, whereas
$\Wreach_r[G,L, v]$ and $\Sreach_r[G,L,v]$ are vertex sets.  The
\emph{$r$-admissibility} $\adm_r(G)$ of~$G$~is
\begin{eqnarray*} 
\adm_r(G) & = & \min_{L\in\Pi(G)}\max_{v\in V(G)}\adm_r[G,L,v].
\end{eqnarray*} 
The generalised colouring numbers were introduced by Kierstead and
Yang \cite{kierstead2003orders} in the context of colouring and marking
games on graphs. The authors also proved that the generalised
colouring numbers are related by the following inequalities:
\begin{equation}\label{eq:gen-col-ineq}
\adm_r(G)\leq \col_r(G)\le \wcol_r(G)\le (\adm_r(G))^r.
\end{equation}

\subparagraph*{Shallow minors, bounded expansion, and nowhere
  denseness.} A graph $H$ with $V(H)=\{v_1,\ldots, v_n\}$ is a
{\em{depth-$r$ minor}} of $G$, denoted $H\minor_rG$, if there is a
minor model $(H_1,\ldots,H_n)$ of $H$ in $G$ such that each $H_i$ has
radius at most $r$. We write $d(H)$ for the \emph{average degree}
of~$H$, that is, for the number $2|E(H)|/|V(H)|$.  A class $\CCC$ of
graphs has \emph{bounded expansion} if there is a function
$f\colon\N\rightarrow\N$ such that for all non-negative integers $r$
we have $d(H)\leq f(r)$ for every $H\minor_r G$ with $G\in\CCC$.  A
class $\CCC$ of graphs is \emph{nowhere dense} if for every real
$\epsilon>0$ and every non-negative integer~$r$, there is an integer
$n_0$ such that if $H$ is an $n$-vertex graph with $n\geq n_0$ and
$H\minor_r G$ for some $G\in\CCC$, then~$d(H)\leq n^\epsilon$.

Bounded expansion and nowhere dense classes of graphs were introduced
by Ne\v{s}et\v{r}il and Ossona de Mendez as models for uniform
sparseness of graphs
\cite{nevsetvril2008grad,nevsetvril2011nowhere}. As proved by
Zhu~\cite{zhu2009colouring}, the generalised colouring numbers are
tightly related to densities of low-depth minors, and hence they can
be used to characterise bounded expansion and nowhere dense classes.

\begin{theorem}[Zhu \cite{zhu2009colouring}] A class $\CCC$ of graphs
  has bounded expansion if and only if there is a function
  $f:\N\rightarrow\N$ such that $\wcol_r(G)\leq f(r)$ for all $r\in\N$
  and all $G\in\CCC$.
\end{theorem}

Due to Inequality~(\ref{eq:gen-col-ineq}), we may equivalently demand that there
is a function $f:\N\rightarrow\N$ such that $\adm_r(G)\leq f(r)$ or
$\col_r(G)\leq f(r)$ for all non-negative integers $r$ and all
$G\in\CCC$.

Similarly, from Zhu's result one can derive a characterisation of
nowhere dense classes of graphs, as presented in
\cite{nevsetvril2011nowhere}.  A class $\CCC$ of graphs is called
\emph{hereditary} if it is closed under induced subgraphs, that is, if
$H$ is an induced subgraph of $G\in\CCC$, then $H\in\CCC$.

\begin{theorem}[Ne\v{s}et\v{r}il and Ossona de Mendez
  \cite{nevsetvril2011nowhere}] A hereditary class $\CCC$ of graphs is
  no\-where dense if and only if for every real $\epsilon>0$ and every
  non-negative integer $r$, there is a positive integer $n_0$ such
  that if $G\in\CCC$ is an $n$-vertex graph with $n\geq n_0$, then
  $\wcol_r(G)\leq n^\epsilon$.
\end{theorem}

As shown in \cite{dvovrak13}, for every non-negative integer $r$,
computing $\adm_r(G)$ is fixed-parameter tractable on any class of
bounded expansion (parameterized by $\adm_r(G)$). For $\col_r(G)$ and
$\wcol_r(G)$ this is not known; however, by~\eqref{eq:gen-col-ineq} we
can use admissibility to obtain approximations of these numbers.  On
nowhere dense classes of graphs, for every $\epsilon>0$ and every
non-negative integer $r$, we can compute an order that witnesses
$\wcol_r(G)\leq n^\epsilon$ in time $\Oof(n^{1+\epsilon})$ if $G$ is
sufficiently large \cite{grohe2014deciding}, based on Ne\v{s}et\v{r}il
and Ossona de Mendez's augmentation technique
\cite{nevsetvril2008grad}.

\section{Uniform quasi-wideness and the splitter game}

In this section we discuss the relation between weak $r$-colouring
numbers and two notions that characterise nowhere dense classes:
uniform quasi-wideness and the splitter game.

For a graph $G$, a vertex subset $A\subseteq V(G)$ is called
\emph{$r$-independent} in $G$, if $\dist_G(a,b)>r$ for all different
$a,b\in V(G)$. A vertex subset is called \emph{$r$-scattered}, if it
is $2r$-independent, that is, if the $r$-neighbourhoods of different
elements of $A$ do not intersect.

Informally, uniform quasi-wideness means the following: in any large
enough subset of vertices of a graph from $\CCC$, one can find a large
subset that is $r$-scattered in $G$, possibly after removing from $G$
a small number of vertices.  Formally, a class $\CCC$ of graphs is
\emph{uniformly quasi-wide} if there are functions
$N:\N\times\N\rightarrow\N$ and $s:\N\rightarrow\N$ such that for all
$m, r\in\N$, if $W\subseteq V(G)$ for a graph $G\in\CCC$ with
$|W|>N(m,r)$, then there is a set $S\subseteq V(G)$ of size at most
$s(r)$ such that $W$ contains a subset of size at least $m$ that is
$r$-scattered in $G-S$.

The notion of quasi-wideness was introduced by
Dawar~\cite{dawar2010homomorphism} in the context of homomorphism
preservation theorems. It was shown in \cite{nevsetvril2010first} that
classes of bounded expansion are uniformly quasi-wide and that uniform
quasi-wideness characterises nowhere dense classes of graphs.

\begin{theorem}[Ne\v{s}et\v{r}il and Ossona de Mendez
  \cite{nevsetvril2010first}] A hereditary class $\CCC$ of graphs is
  no\-where dense if and only if it is uniformly quasi-wide.
\end{theorem}

It was shown by Atserias et al.\@ in \cite{atserias2006preservation} that classes
that exclude $K_k$ 
as a minor are uniformly quasi-wide. In fact, in this case we can choose
$s(r)=k-1$, independent of $r$ (if such a constant function for a class $\CCC$
exists, the class is called \emph{uniformly almost wide}). However, the function
$N(m,r)$ that was used in the proof is huge: it comes from an iterated Ramsey
argument. The same approach was used in \cite{nevsetvril2010first} to show that
every nowhere dense class, and in particular, every class of bounded expansion,
is uniformly quasi-wide.  We present a new proof that every bounded expansion
class is uniformly quasi-wide, which gives us a much better bound on $N(m,r)$
and which is much simpler than the previously known proof.

\begin{theorem} Let $G$ be a graph and let $r,m\in \N$. Let $c\in\N$ be such
  that $\wcol_r(G)\leq c$ and let $A\subseteq V(G)$ be a set of size at least
  $(c+1)\cdot 2^m$. Then there exists a set $S$ of size at most $c(c-1)$ and a set
  $B\subseteq A$ of size at least $m$ which is $r$-independent
  in $G-S$.
\end{theorem}
\begin{proof} Let $L\in \Pi(G)$ be such that
  $|\Wreach_r[G,L,v]|\leq c$ for every $v\in V(G)$. Let $H$ be the
  graph with vertex set $V(G)$, where we put an edge $uv\in E(H)$ if
  and only if $u\in \Wreach_r[G,L,v]$ or $v\in \Wreach_r[G,L,u]$. Then
  $L$ certifies that~$H$ is $c$-degenerate, and hence we can greedily
  find an independent set $I\subseteq A$ of size $2^m$ in $H$. By the
  definition of the graph $H$, we have that
  $\Wreach_r[G,L,v]\cap I=\{v\}$ for each $v\in I$.

\begin{claim}\label{cl:del}
  Let $v\in I$. Then deleting $\Wreach_r[G,L,v]\setminus \{v\}$ from
  $G$ leaves $v$ at a distance greater than $r$ (in
  $G-(\Wreach_r[G,L,v]\setminus\{v\})$) from all the other vertices of
  $I$.
\end{claim}
\begin{claimproof}
  Let $u\in I$ and let $P$ be a path in $G$ that has length at most
  $r$ and connects $u$ and~$v$. Let $z\in V(P)$ be minimal with
  respect to $L$. Then $z<_L v$ or $z=v$. If $z<_L v$, then
  $z\in \Wreach_r[G,L,v]$ and hence the path $P$ no longer exists
  after the deletion of $\Wreach_r[G,L,v]\setminus\{v\}$ from $G$. On
  the other hand, if $z=v$, then $v\in \Wreach_r[G,L,u]$,
  which contradicts the fact that both $u,v\in I$. 
\end{claimproof}

We iteratively find sets
$B_0\subseteq \ldots \subseteq B_m\subseteq I$, sets
$I_0\supseteq \ldots \supseteq I_m$, and sets
$S_0\subseteq \ldots \subseteq S_m$ such that $B$ is $r$-independent
in $G-S$, where $B:=B_m$ and $S:=S_m$. We maintain the invariant that
sets $B_i$, $I_i$, and $S_i$ are pairwise disjoint for each $i$.  Let
$I_0=I$, $B_0 = \emptyset$ and $S_0 = \emptyset$.  In one step
$i=1,2,\ldots, m$, we delete some vertices from $I_i$ (thus obtaining
$I_{i+1}$), shift one vertex from $I_i$ to $B_i$ (obtaining $B_{i+1}$)
and, possibly, add some vertices from $V(G)\setminus I_i$ to $S_i$
(obtaining $S_{i+1}$).  More precisely, let $v$ be the vertex of $I_i$
that is the largest in the order $L$.  We set
$B_{i+1} = B_i\cup\{v\}$, and now we discuss how $I_{i+1}$ and
$S_{i+1}$ are constructed.

We distinguish two cases. First, suppose $v$ is connected by a path of
length at most~$r$ in $G-S_i$ to at most half of the vertices of $I_i$
(including $v$). Then we remove these reachable vertices from~$I_i$,
and set $I_{i+1}$ to be the result. We also set $S_{i+1}=S_i$. Note
that $|I_{i+1}|\geq |I_i|/2$.

Second, suppose $v$ is connected by a path of length at most $r$ in
$G-S_i$ to more than half of the vertices of $I_i$ (including $v$). We
proceed in two steps. First, we add the at most $c-1$ vertices of
$\Wreach_r[G,L,v]\setminus\{v\}$ to $S_{i+1}$, that is, we let
$S_{i+1}=S_i\cup (\Wreach_r[G,L,v]\setminus\{v\})$. (Recall here that
$\Wreach_r[G,L,v]\cap I = \{v\}$.) By Claim~\ref{cl:del}, this leaves
$v$ at a distance greater than~$r$ from every other vertex of $I_i$ in
$G-S_{i+1}$. Second, we construct $I_{i+1}$ from $I_i$ by removing the
vertex $v$ and all the vertices of $I_i$ that are not connected to $v$
by a path of length at most~$r$ in $G-S_i$, hence we have
$|I_{i+1}|\geq
\lfloor|I_i|/2\rfloor$.

Observe the construction above can be carried out for $m$ steps,
because in each step, we remove at most half of the vertices of $I_i$
(rounded up) when constructing $I_{i+1}$.  As $|I_0|=|I|=2^m$, it is
easy to see that the set $I_i$ cannot become empty within $m$
iterations.  Moreover, it is clear from the construction that we end
up with a set $B=B_m$ that has size~$m$ and is $r$-scattered in $G-S$,
where $S=S_m$.  It remains to argue that $|S_m|\leq c(c-1)$. For this,
it suffices to show that the second case cannot apply more than $c$
times in total.

Suppose the second case was applied in the $i$th iteration, when
considering a vertex~$v$.  Every vertex $u\in I_i$ with $u<_Lv$ that
was connected to~$v$ by a path of length at most $r$ in $G-S_i$
satisfies $\Wreach_r[G,L,v]\cap \Wreach_r[G,L,u]\neq \emptyset$.
Thus, every remaining vertex~$u\in I_{i+1}$ has at least one of its
weakly $r$-reachable vertices deleted (that is, included in
$S_{i+1}$).  As the number of such vertices is at most $c-1$ at the
beginning, and it can only decrease during the construction, this
implies that the second case can occur at most~$c$~times.
\end{proof}

As shown in \cite{siebertz16}, if $K_k\not\minor G$, then
$\wcol_r(G)\in\Oof(r^{k-1})$. Hence, for such graphs we have to delete
only a polynomial (in $r$) number of vertices in order to find an
$r$-independent set of size $m$ in a set of vertices of size single
exponential in $m$.

We now implement the same idea to find a very simple strategy for
splitter in the splitter game, introduced by Grohe et
al.~\cite{grohe2014deciding} to characterise nowhere dense classes of
graphs.  Let~${\ell},r\in \N$. The \emph{simple $\ell$-round
  radius-$r$ splitter game} on~$G$ is played by two players,
\emph{connector} and \emph{splitter}, as follows. We let~$G_0:=G$. In
round~$i+1$ of the game, connector chooses a
vertex~$v_{i+1}\in V(G_i)$. Then splitter picks a vertex
$w_{i+1}\in N_r^{G_i}(v_{i+1})$. We
let~$G_{i+1}:=G_i[N_r^{G_i}(v_{i+1})\setminus \{w_{i+1}\}]$. Splitter
wins if~$G_{i+1}=\emptyset$. Otherwise the game continues
at~$G_{i+1}$. If splitter has not won after~${\ell}$ rounds, then
connector wins.

A \emph{strategy} for splitter is a function~$\sigma$ that maps every
partial play $(v_1, w_1, \dots,$ $v_s, w_s)$, with associated
sequence~$G_0, \dots, G_s$ of graphs, and the next
move~$v_{s+1}\in V(G_s)$ of connector, to a
vertex~$w_{s+1}\in N_r^{G_s}(v_{s+1})$ that is the next move of
splitter. A strategy~$\sigma$ is a \emph{winning strategy} for
splitter if splitter wins every play in which she follows the
strategy~$f$. We say that splitter \emph{wins} the simple $\ell$-round
radius-$r$ splitter game on~$G$ if she has a winning strategy.

\begin{theorem}[Grohe et al.\ \cite{grohe2014deciding}] A class $\CCC$
  of graphs is nowhere dense if and only if there is a function
  $\ell:\N\rightarrow\N$ such that splitter wins the simple
  $\ell(r)$-round radius-$r$ splitter game on every graph $G\in\CCC$.
\end{theorem}

More precisely, it was shown in \cite{grohe2014deciding} that
$\ell(r)$ can be chosen as $N(2s(r), r)$, where $N$ and~$s$ are the
functions that characterise $\CCC$ as a uniformly quasi-wide class of
graphs. We present a proof that on bounded expansion classes, splitter
can win much faster.

\begin{theorem}\label{thm:splitterwcol} Let $G$ be a graph, let $r\in\N$ and let
  $\ell=\wcol_{2r}(G)$. Then splitter wins the $\ell$-round radius-$r$
  splitter game.
\end{theorem}
\begin{proof} Let $L$ be a linear order that witnesses
  $\wcol_{2r}(G)=\ell$. Suppose in round $i+1\leq \ell$, connector
  chooses a vertex $v_{i+1}\in V(G_i)$. Let $w_{i+1}$ (splitter's
  choice) be the minimum vertex of $N_r^{G_i}(v_{i+1})$ with respect
  to $L$. Then for each $u\in N_r^{G_i}(v_{i+1})$ there is a path
  between $u$ and $w_{i+1}$ of length at most $2r$ that uses only
  vertices of $N_r^{G_i}(v_{i+1})$. As $w_i$ is minimum in
  $N_r^{G_i}(v_{i+1})$, $w_{i+1}$ is weakly $2r$-reachable from each
  $u\in N_r^{G_i}(v_{i+1})$. Now let
  $G_{i+1}:=G_i[N_r^{G_i}(v_{i+1})\setminus\{w_{i+1}\}]$. As $w_{i+1}$
  is not part of $G_{i+1}$, in the next round splitter will choose
  another vertex which is weakly $2r$-reachable from every vertex of
  the remaining $r$-neighbourhood. As $\wcol_{2r}(G)= \ell$, the game
  must stop after at most $\ell$ rounds.
\end{proof}

\section{Uniform orders for graphs excluding a topological minor}\label{sec:uniform}

If $\CCC$ is a class of bounded expansion such that
$\wcol_r(G)\leq f(r)$ for all $G\in\CCC$ and all $r\in \N$, the order
$L$ that witnesses this inequality for $G$ may depend on the value
$r$. We say that a class $\CCC$ \emph{admits uniform orders} if there
is a function $f:\N\rightarrow\N$ such that for each $G\in\CCC$, there
is a linear order $L\in \Pi(G)$ such that
$|\Wreach_r[G,L,v]|\leq f(r)$ for all $v\in V(G)$ and all $r\in\N$. In
other words, there is one order that simultaneously certifies the
inequality $\wcol_r(G)\leq f(r)$ for all $r$.

It is implicit in \cite{siebertz16} that every class that excludes a
fixed minor admits uniform orders, which can be efficiently
computed. We are going to show that the same holds for classes that
exclude a fixed topological minor. Our construction is similar to the
construction of \cite{siebertz16}, in particular, our orders can be
computed quickly in a greedy fashion. The proof that we find an order
of high quality is based on the decomposition theorem for graphs with
excluded topological minors, due to Grohe and
Marx~\cite{grohe2015structure}.  Note however, that for the
construction of the order we do not have to construct a tree
decomposition according to Grohe and Marx~\cite{grohe2015structure}.

\subparagraph*{Construction.} Let $G$ be a graph. We present a
construction of an order of $V(G)$ of high quality.  We iteratively
construct a sequence $H_1,\ldots, H_\ell$ of pairwise disjoint and
connected subgraphs of $G$ such that
$\bigcup_{1\leq i\leq \ell}V(H_i)=V(G)$. For $0\leq i<\ell$, let
$G_i \coloneqq G - \bigcup_{1\le j\le i}V(H_j)$. We say that a
component $C$ of $G_i$ is \emph{connected} to a subgraph $H_j$,
$j\leq i$, if there is a vertex $u\in V(H_j)$ and a vertex $v\in V(C)$
such that $uv\in E(G)$. For all~$i$, $1\leq i<\ell$, we will maintain
the following invariant. If $C$ is a component of $G_i$, then the
subgraphs $H_{i_1},\ldots, H_{i_s}\in \{H_1,\ldots, H_i\}$ that are
connected to $C$ form a minor model of the complete graph $K_s$, where
$s$ is their number.

To start, we choose an arbitrary vertex $v\in V(G)$ and let $H_1$ be
the connected subgraph $G[\{v\}]$. Clearly, $H_1$ satisfies the above
invariant.  Now assume that for some $i$, $1\le i< \ell$, the sequence
$H_1,\ldots,H_i$ has already been constructed. Fix some component $C$
of $G_i$ and, by the invariant, assume that the subgraphs
$H_{i_1},\ldots, H_{i_s} \in \{H_1,\ldots, H_i\}$ with
$1\leq i_1<\ldots<i_s\leq i$ that have a connection to $C$ form a
minor model of~$K_s$. For a vertex $v\in V(C)$, let $m(v)$ be the
maximum cardinality of a family $\mathcal{P}$ of paths with the
following properties: each path of~$\mathcal{P}$ connects $v$ with a
different subgraph $H_{i_j}$, the internal vertices of each path 
from~$\mathcal{P}$ belong to $G_i$, and the paths of~$\mathcal{P}$ are
pairwise disjoint apart from sharing $v$. Note that $m(v)$ can be
computed in polynomial time using any maximum flow algorithm.
Pick~$v$ to be a vertex of $C$ with maximum $m(v)$.  Let $T$ be the
tree of the breadth-first search in $G[C]$ that starts in $v$; thus,
$T$ is rooted at~$v$.  We choose $H_{i+1}$ to be a minimal connected
subtree of $T$ that contains $v$ and, for each~$j$ with
$1\leq j\leq s$, at least one neighbour of $H_{i_j}$ in $C$.

From the construction it is easy to see that for every component $C'$
of $G_{i+1}$, the subgraphs
$H'_{i_1},\ldots, H'_{i_{s'}}\in \{H_1,\ldots, H_{i+1}\}$ that are
connected to $C'$ form the minor model of a complete graph, hence the
invariant is again established.  Having chosen $H_{i+1}$, we proceed
to the next iteration. The construction stops when all vertices are
part of some $H_i$, $1\leq i\leq \ell$.

We construct an order $L$ of $V(G)$ as follows. Let $v<_Lu$ if
$v\in V(H_i)$ and $u\in V(H_j)$ for some $i<j$. Furthermore, we order
the vertices within each $H_i$ arbitrarily.  Obviously, the
construction does not depend on~$r$, hence the produced order is
uniform for~$G$.

\subparagraph*{Analysis.} From now on we assume that $G$ excludes $K_k$
as a topological minor, for some constant~$k$. Furthermore, assume
that the graphs $H_1,\ldots, H_\ell$ and a corresponding order $L$
have been constructed, as described above. We now show that the
constructed order has good qualities. Our proof is based on the
following two key lemmas. The first lemma states that for every
component $C$ of $G_i$ arising after the construction of
$H_1,\ldots, H_i$, every vertex $v$ of $C$ can reach only a bounded
number of subgraphs among $H_1,\ldots, H_i$ by disjoint paths.

\begin{lemma}\label{lem:disjointpaths}
  There is a constant $\alpha$ (depending only on $k$) such that for
  all integers $i$, $1\leq i<\ell$, if $C$ is a component of $G_i$,
  then for every vertex $v\in V(C)$, we have $m(v)\leq \alpha$, where
  $m(v)$ is defined as in the construction.
\end{lemma}

The second lemma states that from a vertex of $H_{i+1}$, we can reach
only a bounded number of vertices of each $H_j$, $1\leq j\leq i+1$, by
short disjoint paths in $G_i$.

\begin{lemma}\label{lem:reachablevertices}
  There is a constant $\beta$ (depending only on $k$) such that for
  all integers $i,j$, where $1\leq j\leq i\leq\ell$, and all positive
  integers~$r$, the following holds.  Suppose $v\in V(H_i)$, and let
  $\mathcal{P}$ be any family of paths of length at most $r$ with the
  following properties: each path from~$\mathcal{P}$ connects~$v$ with
  a different vertex of $H_j$, the internal vertices of $\mathcal{P}$
  belong to $G_j$, and paths from $\mathcal{P}$ are internally vertex
  disjoint.  Then $\mathcal{P}$ has size not larger than
  $\beta\cdot r$.
\end{lemma}

It is easy to show that the above two lemmas guarantee that $L$ has
the required properties.  

\begin{corollary}\label{cor:wrapup}
  If $K_k\not\minor^tG$, then there exists a constant $\gamma$
  (depending only on $k$) and a uniform order $L$ that witnesses
  $\adm_r(G)\leq \gamma\cdot r$ for all non-negative integers $r$.
\end{corollary}
\begin{proof}
  The $r$-admissibility of a vertex $v$ is determined in the construction when 
  $v$ first appears in some $H_i$, $1\leq i<\ell$. More precisely, $\adm_r(G)$ is 
  upper bounded by the maximum possible number of disjoint paths in $G_i$ of 
  length at most $r$ from $v$ to 
  vertices of $H_1\cup \ldots\cup H_i$. By Lemma~\ref{lem:disjointpaths},
  there is a constant $\alpha$ such that $v$ can reach at most $\alpha$ distinct
  $H_j$, $1\leq j< i$, via internally vertex disjoint paths in $G_i$. Additionally, $v$ can 
  reach other vertices of $H_i$. By Lemma~\ref{lem:reachablevertices}, there is a constant $\beta$ such that $v$ 
  reach at most $\beta\cdot r$ vertices of each $H_j$ with $j\leq i$ by internally
  vertex disjoint paths of length at most~$r$. Let
  $\gamma\coloneqq (\alpha+1)\cdot \beta$, then $v$ can reach at most $\gamma\cdot r$ smaller
  vertices by internally vertex disjoint paths of length at most $r$ whose
  internal vertices are larger than~$v$ w.r.t.\ $L$.
\end{proof}

The proof of Lemma \ref{lem:disjointpaths} is based on the
decomposition theorem for graphs with excluded topological minors of
Grohe and Marx \cite{grohe2015structure}. Recall that a \emph{tree
  decomposition} of a graph~$G$ is a pair $(T,\beta)$, where $T$ is a
tree and $\beta:V(T)\rightarrow 2^{V(G)}$, such that for every vertex
$v\in V(G)$ the set $\beta^{-1}(v)=\{t\in V(T) : v\in \beta(t)\}$ is
non-empty and connected in $T$, and for every edge $e\in E(G)$ there
is a node $t\in V(T)$ such that $e\subseteq \beta(t)$. The
\emph{width} of $(T,\beta)$ is $\max\{|\beta(t)|-1 : t\in V(T)\}$ and
the \emph{adhesion} of $(T,\beta)$ is
$\max\{|\beta(s)\cap \beta(t)| : st\in E(T)\}$.

For a node $t\in T$, we call $\beta(t)$ the \emph{bag at $t$}. If
$T'\subseteq T$, we write $\beta(T')$ for
$\bigcup_{t'\in V(T')}\beta(t')$ and if $M\subseteq V(G)$, we write
$\beta^{-1}(M)$ for $\bigcup_{v\in M}\beta^{-1}(v)$. Denote by $K[X]$
the complete graph on a vertex set $X$. The \emph{torso at $t$} is the
graph
$\tau(t):=G[\beta(t)]\cup \bigcup_{st\in E(T)}K[\beta(s)\cap
\beta(t)]$.

\begin{theorem}[\cite{grohe2015structure}]\label{thm:decomposition}
  For every $k\in\N$, there exist constants $a(k),c(k),d(k)$ and
  $e(k)$ such that the following holds. Let $H$ be a graph on $k$
  vertices. Then for every graph $G$ with $H\not\minor^tG$ there is a
  tree decomposition $(T,\beta)$ of adhesion at most $a(k)$ such that
  for all $t\in V(T)$ one of the following two alternatives hold.
  \smallskip
\begin{enumerate}
\item The torso $\tau(t)$ has at most $c(k)$ vertices of degree larger
  than $d(k)$, which we call the \emph{apex vertices} of
  $\tau(t)$. Such a node $t$ will be called a \emph{bounded degree
    node}.
\item The torso $\tau(t)$ excludes the complete graph $K_{e(k)}$ as a
  minor. Such a node $t$ will be called an \emph{excluded minor node}.
\end{enumerate}
\end{theorem}

We will need the following well-known properties of trees and tree
decompositions.

\begin{lemma}[Helly-property for trees]\label{lem:helly}
  Let $T$ be a tree and let $(T_i)_{i\in I}$ be a family of subtrees
  of $T$. If $V(T_i) \cap V(T_j)\not=\emptyset$, for all $i,j\in I$,
  then $\bigcap_{i\in I} V(T_i) \not= \emptyset$.
\end{lemma}

\begin{lemma}\label{lem:separator}
  Let $(T,\beta)$ be a tree decomposition of a graph $G$. Let $e=st$
  be an edge of $T$ and let $T_1,T_2$ be the components of $T-e$.
  Then $\beta(s)\cap \beta(t)$ separates $\beta(T_1)$ from
  $\beta(T_2)$, that is, every path from a vertex of $\beta(T_1)$ to a
  vertex of $\beta(T_2)$ traverses a vertex of
  $\beta(s)\cap \beta(t)$.
\end{lemma}

\begin{lemma}\label{lem:connected}
  If $H\subseteq G$ is a connected subgraph of $G$, then
  $\beta^{-1}(V(H))$ is connected in~$T$.
\end{lemma}

For the proof of Lemma \ref{lem:disjointpaths}, assume that $G$ is
decomposed as described by Theorem \ref{thm:decomposition}. Assume
that $H_1,\ldots, H_i$ have been constructed and let $C$ be a
component of $G_i$ that has a connection to the subgraphs
$H_{i_1},\ldots, H_{i_s}$.  Recall that throughout the construction we
guarantee that the subgraphs $H_{i_1},\ldots, H_{i_s}$ form the minor
model of a complete graph $K_s$.  We first identify one bag of the
decomposition as a bag which intersects many distinct branch sets of
this minor model.  The following lemma follows easily from the
separator properties of tree decompositions, in particular Lemma~\ref{lem:separator}.

\begin{lemma}\label{lem:corebag}
  There can be at most one node $t$ such that $\beta(t)$ intersects
  strictly more than $a(k)$ of the branch sets $H_{i_j}$, for
  $1\leq j\leq s$.
\end{lemma}
\begin{proof}
  Assume there are two distinct nodes $t_1$ and $t_2$ with this property, 
  and suppose $\beta(t_1)$ intersects $R^1_1, \ldots, R^1_{a(k)+1}\in \{H_{i_1},\ldots, H_{i_s}\}$, whereas 
  $\beta(t_2)$ intersects $R^2_1, \ldots, R^2_{a(k)+1}\in \{H_{i_1},\ldots, H_{i_s}\}$.
  Note that some branch sets $R^1_p$ and $R^2_q$ may coincide, but by reorder if necessary, we may assume w.l.o.g. that $R^1_p=R^2_q$ can happen only if $p=q$. 
  Observe that since $H_{i_1},\ldots, H_{i_s}$ form a minor model of a
  complete graph, the union of the vertex sets of $R^1_p$ and $R^2_p$ induces a connected subgraph of $G$, for each $p$ with $1\leq p\leq a(k)+1$; note that this is true also if $R^1_p=R^2_p$. 
  Consequently, if $e=s_1s_2$ is any edge of $T$ whose removal disconnects $t_1$ from $t_2$, then by Lemma~\ref{lem:separator} we have that $\beta(s_1)\cap \beta(s_2)$ must contain at least one vertex from 
  $V(R^1_p)\cup V(R^2_p)$ for each $p$ with $1\leq p\leq a(k)+1$. Sets $V(R^1_p)\cup V(R^2_p)$ are disjoint for distinct $p$, hence we conclude that $|\beta(s_1)\cap \beta(s_2)|\geq a(k)+1$.
  This contradicts the fact that $(T,\beta)$ has adhesion at most $a(k)$.
\end{proof}

We now show that there is a bag that intersects every branch set. The
proof is a simple application of the Helly property of trees
(Lemma~\ref{lem:helly}) and Lemma~\ref{lem:connected}.

\begin{lemma}\label{lem:corebag-exists}
  There is a node $t$ such that $\beta(t)$ intersects each $H_{i_j}$,
  for $1\leq j\leq s$.
\end{lemma}
\begin{proof}
As each $H_{i_j}$ is connected, by Lemma~\ref{lem:connected} we have that $\beta^{-1}(V(H_{i_j}))$ 
induces a subtree $T_{i_j}$ in $T$. As the $\{H_{i_1},\ldots, H_{i_s}\}$ form a model of a complete graph, they 
are pairwise connected and hence, by the definition of tree decompositions, any two such subtrees 
intersect. By Lemma~\ref{lem:helly}, they all intersect in one node $t$, which is as desired. 
\end{proof}

Hence, provided $s>a(k)$, there is a node $t$ with $\beta(t)$
intersecting at least $a(k)+1$ branch sets~$H_{i_j}$.  By
Lemma~\ref{lem:corebag}, this node is unique. We call it the
\emph{core node} of the minor model.  Next we show that if the model
is large, then its core node must be a bounded degree node. Shortly
speaking, this is because the model $H_{i_1},\ldots, H_{i_s}$ trimmed
to the torso of the core node is already a minor model of $K_s$ in
this torso.

\begin{lemma}\label{lem:degreebag}
  If $s>\max\{a(k), e(k)\}$, then the core node of the minor model is
  a bounded degree node.
\end{lemma}
\begin{proof}
  As $s>a(k)$, by Lemma~\ref{lem:corebag} we can identify the 
  unique core node $t$ whose bag intersects all the branch sets $H_{i_j}$. 
  Recall that $\tau(t)$ is the graph induced by
  the bag $\beta(t)$ in which all adjacent separators are turned into cliques. 
  It is easy to see that the subgraphs $H_{i_j}':=\tau(t)[V(H_{i_j})\cap \beta(t)]$ are connected
  in $\tau(t)$ and form a minor model of $K_s$. As $s>e(k)$, 
  we infer that $t$ cannot be an excluded minor node, and hence it is a bounded degree node.
\end{proof}

For vertices outside the bag of the core node, the bound promised in
Lemma~\ref{lem:disjointpaths} can be proved similarly as
Lemma~\ref{lem:corebag}.

\begin{lemma}\label{lem:connectionsfromoutside}
  Let $C$ be a component of $G_i$ that has a connection to the
  subgraphs $H_{i_1},\ldots, H_{i_s}$. If $s>a(k)$, then for every
  vertex $v\in V(C)\setminus\beta(t)$, where $t$ is the core node of
  the model, we have that $m(v)\leq a(k)$.
\end{lemma}
\begin{proof}
By the properties of a tree decomposition, there is an edge $e=tt'$ of $T$ such that
$\beta^{-1}(v)$ is contained in the subtree of $T-e$ that contains $t'$.
Suppose $\mathcal{P}$ is a family of paths that connect $v$ with distinct branch sets $H_{i_j}$ and are pairwise disjoint apart from $v$.
Recall that $\beta(t)$ intersects every branch set $H_{i_j}$.
Therefore, by extending each path of $\mathcal{P}$ within the branch set it leads to, we can assume w.l.o.g. that each path of $\mathcal{P}$ connects $v$ with a vertex of $\beta(t)$.
By Lemma~\ref{lem:separator}, this implies that each path of $\mathcal{P}$ intersects $\beta(t)\cap \beta(t')$.
Paths of $\mathcal{P}$ share only $v$, which is not contained in $\beta(t)\cap \beta(t')$, and hence we conclude that $|\mathcal{P}|\leq |\beta(t)\cap \beta(t')|$.
As $\mathcal{P}$ was chosen arbitrarily, we obtain that $m(v)\leq |\beta(t)\cap \beta(t')|\leq a(k)$.
\end{proof}

\vspace{-0.5cm}
We now complete the proof of Lemma~\ref{lem:disjointpaths} by looking
at the vertices inside the core bag.

\begin{proof}[Proof of Lemma~\ref{lem:disjointpaths}]
  We set $\alpha:=a(k)+c(k)+d(k)+e(k)$.  Assume towards a
  contradiction that for some $i$, $1\leq i<\ell$, we have that some
  component $C$ of $G_i$ contains a vertex $v_1$ with $m(v_1)>\alpha$.
  Denote the branch sets that have a connection to $C$ by
  $H_{i_1},\ldots, H_{i_s}$, where $i_1<i_2<\ldots<i_s$.  Let
  $\mathcal{P}$ be a maximum-size family of paths that pairwise share
  only $v_1$ and connect $v_1$ with different branch sets $H_{i_j}$.
  As $m(v_1)>\alpha$, we have that $|\mathcal{P}|>\alpha$, and in
  particular $s>\alpha$.  As $\alpha>a(k)$, by~Lemma~\ref{lem:corebag} 
  and~Lemma~\ref{lem:corebag-exists} we can
  identify the unique core node~$t$ of the minor model. As
  $s>\max\{a(k),e(k)\}$, by Lemma~\ref{lem:degreebag} the core node is
  a bounded degree node.  As $m(v_1)>a(k)$, by
  Lemma~\ref{lem:connectionsfromoutside} we have $v_1\in \beta(t)$.
  As $\mathcal{P}$ contains more than $d(k)$ disjoint paths from $v$
  to distinct branch sets, the degree of $v_1$ in $G$ must be greater
  than $d(k)$, hence $v_1$ is an apex vertex of $\tau(t)$.

  Since $i_1<i_2<\ldots<i_s$, we have that the component $C$ was
  created when $H_{i_s}$ was removed from $G_{i_s-1}$. Let $C'$ be the
  component of $G_{i_s-1}$ that contains $C$ and $H_{i_s}$ (and thus
  $v_1$).  Observe that $C'$ is still connected to
  $H_1, \ldots, H_{i_{s-1}}$, and possibly to some other branch sets.
  Recall that $H_{i_s}$ was constructed as a subtree of the
  breadth-first search tree in $G_{i_s}$ that started in a vertex
  $v_2\in V(C')$ which, at this point of the construction, had maximum
  $m(v_2)$ among vertices in $C'$.  However, at this point vertex
  $v_1$ was also present in $C'$, and $\mathcal{P}$ certifies that it
  could send at least $\alpha-1$ disjoint paths to different branch
  sets among $H_1, \ldots, H_{i_{s-1}}$ (in $\mathcal{P}$, at most one
  path leads to $H_{i_s}$, and all the other paths are also present in
  $C'$). We infer that it held that $m(v_2)\geq \alpha-1$ at the
  moment $v_2$ was taken.  Since
  $\alpha>a(k)+c(k)+d(k)+e(k)\geq a(k)+d(k)+e(k)+1$, the same
  reasoning as above shows that $t$ is also the core vertex of the
  minor model formed by branch sets connected to $C'$.  Thus, by
  exactly the same reasoning we obtain that $v_2$ is also an apex
  vertex of $\tau(t)$.

  Since $\alpha>a(k)+c(k)+d(k)+e(k)$, we can repeat this reasoning
  $c(k)+1$ times, obtaining vertices $v_1,\ldots, v_{c(k)+1}$, which
  are all apex vertices of $\tau(t)$.  This contradicts the fact that
  $\tau(t)$ contains at most $c(k)$ apex vertices.
\end{proof}


\begin{proof}[Proof of Lemma~\ref{lem:reachablevertices}]
  We set $\beta$ so that $\beta\cdot r\geq (2r+1)\cdot \alpha$, where
  $\alpha$ is the constant given by Lemma~\ref{lem:disjointpaths}.
  For the sake of contradiction, suppose there is a family of paths
  $\mathcal{P}$ as in the statement, whose size is larger than
  $(2r+1)\cdot \alpha$.

  Recall that $H_{j}$ was chosen as a subtree of a breadth-first
  search tree in $G_{j-1}$; throughout the proof, we treat $H_j$ as a
  rooted tree.  As $H_j$ is a subtree of a BFS tree, every path from a
  vertex $w$ of the tree to the root $v'$ of the tree is an isometric
  path in $G_{j-1}$, that is, a shortest path between $w$ and $v'$ in
  the graph $G_{j-1}$.  If $P$ is an isometric path in a graph $H$,
  then $|N_r^H(v)\cap V(P)|\leq 2r+1$ for all $v\in V(H)$ and all
  $r\in\N$.  As the paths from $\mathcal{P}$ are all contained in
  $G_{j-1}$, and they have lengths at most $r$, this implies that the
  path family $\mathcal{P}$ cannot connect $v$ with more than $2r+1$
  vertices of $H_{j}$ which lie on the same root-to-leaf path in
  $H_{j}$. Since $|\mathcal{P}|> (2r+1)\cdot \alpha$, we can find a
  set $X\subseteq V(H_j)$ such that $|X|>\alpha$, each vertex of $X$
  is connected to $v$ by some path from $\mathcal{P}$, and no two
  vertices of $X$ lie on the same root-to-leaf path in $H_j$.  Recall
  that, by the construction, each leaf of $H_j$ is connected to a
  different branch set $H_{j'}$ for some $j'<j$.  Consequently, we can
  take the paths of $\mathcal{P}$ leading to $X$ and extend them
  within $H_j$ to obtain a family of more than $\alpha$ disjoint paths
  in $G_{j-1}$ that connect $v$ with different branch sets $H_{j'}$
  for $j'<j$.  This contradicts Lemma~\ref{lem:disjointpaths}.
\end{proof}

\enlargethispage{1cm}
\vspace{-0.5cm}
Observe that the order can be computed in time $\Oof(n^5)$:
for each vertex, we compute by a standard flow algorithm in time
$\Oof(n^3)$ whether it should be chosen as the next tree root to form
a subgraph $H_{i_j}$. This choice has to be made at most $n$ times.

Finally, we state one property of the construction that follows
immediately from Lemma~\ref{lem:disjointpaths}.

\begin{lemma}\label{lem:degbound}
  Each constructed subgraph $H_i$ has maximum degree at most
  $\alpha+1$, where $\alpha$ is the constant given by
  Lemma~\ref{lem:disjointpaths}.
\end{lemma}

\section{Model-checking for successor-invariant first-order
formulas}\label{sec:mc}

A \emph{finite and purely relational signature} $\tau$ is a finite set
$\{R_1,\ldots, R_k\}$ of relation symbols, where each relation symbol
$R_i$ has an associated arity $a_i$. A finite $\tau$-structure $\strA$
consists of a finite set~$A$, the universe of $\strA$, and a relation
$R_i(\strA)\subseteq A^{a_i}$ for each relation symbol $R_i\in \tau$.
If $ \strA$ is a finite $\tau$-structure, then the \emph{Gaifman
  graph} of $\strA$, denoted $G(\strA)$, is the graph with
$V(G(\strA))=A$ and there is an edge $uv\in E(G(\strA))$ if and only
if $u\neq v$ and $u$ and $v$ appear together in some relation
$R_i(\strA)$ of $\strA$. We say that a class $\CCC$ of finite
$\tau$-structures has \emph{bounded expansion} if the graph class
$G(\CCC):=\{G(\strA) : \strA\in \CCC\}$ has bounded
expansion. Similarly, for $r\in\N$, we write $\adm_r(\strA)$ for
$\adm_r(G(\strA))$ etc.

Let $V$ be a set. A successor relation on $V$ is a binary relation
$S\subseteq V\times V$ such that $(V,S)$ is a directed path of length
$|V|-1$. Let $\tau$ be a finite relational signature. A formula
$\phi\in \FO[\sigma\cup\{S\}]$ is \emph{successor-invariant} if for
all $\tau$-structures $\strA$ and for all successor relations
$S_1,S_2$ on $V(\strA)$ it holds that
$(\strA, S_1)\models\phi\Longleftrightarrow (\strA,S_2)\models\phi$.

Successor-invariant logics have been studied in database theory and
finite model theory in the past. It was shown by Rossman
\cite{rossman2007successor} that successor-invariant FO is more
expressive than FO without access to a successor relation. It is known
that successor-invariant FO (in fact even order-invariant FO) can
express only local queries \cite{grohe2000locality}, however, the
proof does not translate formulas into local FO-formulas which could
be evaluated algorithmically.  It was shown in
\cite{eickmeyer2013model} that the model-checking problem for
successor-invariant first-order formulas is fixed-parameter tractable
on any proper minor closed class of graphs. Very recently, the same
result was shown for classes with excluded topological 
minors~\cite{eickmeyer2016model}. We give a new proof of the model-checking
result of~\cite{eickmeyer2016model} which is based on the nice
properties of the order we have constructed for graphs that exclude a
topological minor.

Eickmeyer et al.~\cite{eickmeyer2013model} showed that on well-behaved
classes of graphs one can apply the following reduction from the
model-checking problem for successor-invariant formulas to the
model-checking problem for plain first-order formulas.

\begin{lemma}[Eickmeyer et al.~\cite{eickmeyer2013model}]\label{thm:eickmeyer-combine}
  Let $\CCC$ be a class of $\tau$-structures such that for each
  $\strA\in\CCC$ one can compute in polynomial time a graph $H(\strA)$
  such that
\begin{enumerate}
\item $V(H(\strA))=V(G(\strA))$ and
  $E(H(\strA))\supseteq E(G(\strA))$.
\item $H$ contains a spanning tree $T$ which can be computed in
  polynomial time and which is of maximum degree $d$ for some fixed
  integer $d$ depending on $\CCC$ only.
\item The model-checking problem for first-order formulas on the graph
  class $\{H(\strA) : \strA\in\CCC\}$ is fixed-parameter tractable.
\end{enumerate}
Then the model-checking problem for successor-invariant first-order
formulas is fixed-pa\-ra\-me\-ter tractable on $\CCC$.
\end{lemma}

We remark that the original lemma from \cite{eickmeyer2013model}
refers to $k$-walks in $H$, which are easily seen to be equivalent to
spanning trees of maximum degree $k$. In our view, spanning trees are
more intuitive to handle in our graph theoretic context.

\begin{lemma}\label{lem:constructingH}
  Let $k\in\N$.  There is a constant $\delta$, depending only on $k$,
  and a function $f\colon \N\rightarrow \N$ such that the following
  holds.  For every graph $G$ with $K_k\not\minor^tG$ we can compute
  in polynomial time a supergraph $H$ with $V(H)=V(G)$ and
  $E(H)\supseteq E(G)$ such that $\adm_r(H)\leq f(r)$ for all
  $r\in \N$ and such that $H$ contains a spanning tree $T$ with
  maximum degree at most $\delta$; furthermore, such a spanning tree
  $T$ can be also computed in polynomial~time.
\end{lemma}
\begin{proof}
  Without loss of generality, we assume that $G$ is
  connected. Otherwise, we may apply the construction in each
  connected component separately, and then connect the components
  arbitrarily using single edges (added to $H$) in a path-like manner.
  It is easy to see that including the additional edges to the
  spanning tree increases its maximum degree by at most~$2$, while the
  admissibility of the graph also increases by at most $2$.

  We perform the construction of the subgraphs $H_1,\ldots, H_\ell$
  almost exactly as in~Lemma~\ref{sec:uniform}.  However, when
  constructing the $H_i$'s and the order $L$, we put some additional
  restrictions that do not change the quality of $L$.  First, recall
  that when we defined $H_{i+1}$, for some $0\leq i<\ell$, we
  considered a tree of breadth-first search starting at $v_{i+1}$ in a
  connected component $C$ of $G_i$. Suppose that the subgraphs that
  $C$ is connected to are $H_{i_1},\ldots,H_{i_s}$, where
  $1\leq i_1<\ldots<i_s\leq i$.  Then $H_{i+1}$ was defined as a
  minimal subtree of the considered BFS tree that contained, for each
  $1\leq j\leq s$, some vertex of $H_{i_j}$ that is adjacent to $C$.
  Observe that in the construction we were free to choose which
  neighbour of $H_{i_j}$ will be picked to be included in $H_{i+1}$.
  For $j<s$ we make an arbitrary choice as before, but the neighbour
  of $H_{i_s}$ (if exists; note that this is the case for $i>0$) is
  chosen as follows.  We first select the vertex
  $w'_{i+1}\in V(H_{i_s})$ that is the largest in the order $L$ among
  those vertices of $H_{i_s}$ that are adjacent to $C$ (the vertices
  of $H_j$ for $j\leq i$ are already ordered by $L$ at this point).
  Then, we select any its neighbour $w_{i+1}$ in $C$ as the vertex
  that is going to be included in $H_{i+1}$ in its construction.
  Finally, recall that in the construction of $L$, we could order the
  vertices of $H_{i+1}$ arbitrarily.  Hence, we fix an order of
  $H_{i+1}$ so that $w_{i+1}$ is the smallest among $V(H_{i+1})$.
  This concludes the description of the restrictions applied to the
  construction.

  We now construct $H$ by taking $G$ and adding some edges. During the
  construction, we will mark some edges of $H$ as {\em{spanning
      edges}}.  We start by marking all the edges of all the trees
  $H_i$, for $1\leq i\leq \ell$, as spanning edges. At the end, we
  will argue that the spanning edges form a spanning tree of $H$ with
  maximum degree at most $\delta$.

  For each $i$ with $1\leq i<\ell$, let us examine the vertex
  $w_{i+1}$, and let us {\em{charge}} it to $w_{i+1}'$.  Note that in
  this manner every vertex $w_{i+1}$ is charged to its neighbour that
  lies before it in the order $L$.  For any $w\in V(G)$, let $D(w)$ be
  the set of vertices charged to $w$. Now examine the vertices of $G$
  one by one, and for each $w\in V(G)$ do the following. If
  $D(w)=\emptyset$, do nothing. Otherwise, if
  $D(w)=\{u_1,u_2,\ldots,u_h\}$, mark the edge $wu_1$ as a spanning
  edge, and add edges $u_1u_2,u_2u_3,\ldots,u_{h-1}u_h$ to $H$,
  marking them as spanning edges as well.

\begin{claim}\label{cl:degbnd}
  The spanning edges form a spanning tree of $H$ of maximum degree at
  most $\alpha+4$, where~$\alpha$ is the constant given by
  Lemma~\ref{lem:disjointpaths}.
\end{claim}
\begin{claimproof}
Because the branch sets partition the graph, the spanning edges form a spanning subgraph of $H$.
Because we connect the branch set $H_{i+1}$ only to the largest reachable branch set $H_{i_s}$ 
(and this set is never again the largest reachable branch set for $H_j$, $j>i$), the spanning 
subgraph is acyclic. It is easy to see that the spanning subgraph is also connected.  
By Lemma~\ref{lem:degbound}, we have that each $H_i$ has maximum degree at most $\alpha+1$.
Also, for every vertex $w\in V(G)$, at most $3$ additional edges incident to $w$ in $H$ are marked as spanning (two edges are contributed by the path from $u_1$ to $u_h$ (only $u_1$ charges to a different vertex and has degree $1$ on the path) and one edge 
may be added if a vertex is charged to it). In total, this means that
$H$ has maximum degree bounded by $\alpha+4$.
\end{claimproof}

It remains to argue that $H$ has small admissibility. For this, it
suffices to prove the following claim.  The proof uses the additional
restrictions we introduced in the construction.

\begin{claim}\label{cl:adm}
  Let $r$ be a positive integer. If the order $L$ satisfies 
  $\max_{v\in V(G)} |\Sreach_{2r}[G,L,v]|\leq m$, that is, 
  the order certifies $\col_{2r}(G)\leq m$, 
  then $\adm_r(H)\leq m+2$.
\end{claim}
\begin{claimproof}
We verify that for each $r$, the order $L$ certifies that $\adm_r(H)\leq m+2$. For this, take any vertex $v\in V(H)=V(G)$, and let $\mathcal{P}$ be any family of paths of length at most $r$ in $H$
that start in $v$, end in distinct vertices smaller than $v$ in $L$, and are pairwise internally disjoint.
We can further assume that all the internal vertices of all the paths from $\mathcal{P}$ are larger than $v$ in $L$.
Let $i$, $0\leq i<\ell$, be such that $v\in V(H_{i+1})$.
We distinguish two cases: either $v=w_{i+1}$ or $v\neq w_{i+1}$.

We first consider the case $v\neq w_{i+1}$; the second one will be very similar. By the construction of the order $L$, it follows that $w_{i+1}<_L v$.
Consider any path $P\in \mathcal{P}$. Then $P$ is a path in $H$; we shall modify it to a walk $P'$ in $G$ as follows.
Suppose $P$ uses some edge $e$ that is not present in~$H$. By the construction of~$H$, it follows that
$e=u_1u_2$ is an edge connecting two vertices that are charged to the same vertex $w$; suppose w.l.o.g. $P$ traverses $e$ from $u_1$ to $u_2$. 
Define $P'$ by replacing the traversal of $e$ on $P$ by a path of length two consisting of $u_1w$ and $wu_2$, and making the same replacement for all other edges on $P$ that do not belong to $G$.

We claim that all the internal vertices of $P'$ are not smaller, in $L$, than $v$. For this, it suffices to show that whenever some edge $u_1u_2$ is replaced by a path $(u_1w, wu_2)$ as above,
then we have that $v\leq_L w$. Aiming towards a contradiction, suppose that $u_1u_2$ is the first edge on $P$ for which we have $w<_L v$.
By the construction, it must be that $(u_1,u_2)=(w_{j_1},w_{j_2})$ for some $j_1,j_2>i+1$, and $w=w_{j_1}'=w_{j_2}'$. 
Let $j$ be such that $w\in H_j$.
When constructing $H_{j_1}$, we chose $w=w_{j_1}'$ as the largest, w.r.t. $L$, vertex of $H_j$ which was adjacent to the connected component $C'$ of $G_{j_1-1}$ that contains $H_{j_1}$.
Observe that the prefix of $P'$ up to $w_{j_1}$ is a path in $G$ that, by the choice of $u_1u_2$, contains only vertices not smaller in $L$ than $v$. This prefix has to access the connected component $C'$
from some vertex $q$, for which we of course have $v\leq_L q$. If $q\notin V(H_j)$ then, as $H_j$ is the last among subgraphs connected to $C'$, we have that $q\in V(H_{j'})$ for some $j'<j$ and, consequently,
$v\leq_L q<_L w$.
Otherwise, if $q\in V(H_j)$, then by the choice of $w=w_{j_1}'$ as the last, in $L$, vertex of the neighbourhood of $C'$ within $H_j$, we also have $v\leq_L q\leq_L w$.
In both cases we conclude that $v\leq_L w$, a contradiction.

Hence, if we apply the above procedure to all the paths from $\mathcal{P}$, we obtain a family $\mathcal{P}'$ of walks in $G$ with the following properties: 
each walk of $\mathcal{P}'$ has length at most $2r$, it connects $v$ with a different vertex that is smaller in $L$ than $v$, and all its internal vertices are not smaller than $v$.
Note here that walks from $\mathcal{P}'$ are not necessarily disjoint, but still their number must be bounded by $\max_{v\in V(G)} |\Sreach_{2r}[G,L,v]|\leq m$. It follows that $|\mathcal{P}|\leq m$.

Finally, we consider the second case $v=w_{i+1}$. Observe that in the construction of $H$, we added at most $2$ additional edges incident to $v$ that connect $v$ with other vertices charged to the same vertex.
At most two paths from $\mathcal{P}$ can use these edges, and for the other paths we may apply exactly the same reasoning as in the first case. 
It follows that $|\mathcal{P}|\leq m+2$ in this case; this concludes the proof.
\end{claimproof}

The statement of the lemma now directly follows from Claim~\ref{cl:degbnd} and Claim~\ref{cl:adm}.
\end{proof}

Given a graph $G$ that excludes $K_k$ as a topological minor, let us
write $H(G)$ for a graph constructed according to
Lemma~\ref{lem:constructingH}.

\begin{corollary}
The class $\{H(G) : K_k\not\minor^t G\}$ has bounded expansion. 
\end{corollary}

We can now use Theorem~\ref{thm:eickmeyer-combine} to combine the
following result of Dvo\v{r}ak et al.~\cite{dvovrak2013testing} with
Lemma~\ref{lem:constructingH}, to prove fixed-parameter tractability
of successor-invariant $\FO$ on classes that exclude a fixed
topological minor.

\begin{lemma}[Dvo\v{r}\'ak et al.~\cite{dvovrak2013testing}]
  The model-checking problem for first-order formulas is
  fixed-parameter tractable on any class of bounded expansion.
\end{lemma}

\begin{corollary}
  The model-checking problem for successor-invariant first-order
  formulas is fixed parameter tractable on any class of graphs that
  excludes a fixed topological minor.
\end{corollary}

\section{Conclusions}

In this work we gave several new applications of the generalised
colouring numbers on classes of bounded expansion.  In particular, we
have shown that whenever a graph class $\CCC$ excludes some fixed
topological minor, then any graph from $\CCC$ admits one ordering of
vertices that certifies the boundedness of the generalised colouring
numbers for all radii $r$ at once.  It is tempting to conjecture that
such an ordering exists for any graph class of bounded expansion.

Our construction of the uniform ordering proved to be useful in
showing that model-checking successor-invariant $\FO$ is FPT on any
graph class that excludes a fixed topological minor. We believe that
our construction may be helpful in extending this result to any graph
class of bounded expansion, since both the construction of the order,
and the reasoning of Section~\ref{sec:mc}, are oblivious to the fact
that the graph class excludes some topological minor.  The only place
where we used this assumption is the analysis of the constructed
order.

\pagebreak
\bibliographystyle{plainurl}
\bibliography{ref}

\end{document}